\documentclass[11pt]{llncs}

\usepackage{microtype}

\usepackage{wrapfig}
\usepackage{microtype}
\usepackage{graphicx}
\usepackage{graphics}
\usepackage{epsfig}
\usepackage{epstopdf}
\usepackage{amsmath}
\usepackage{latexsym}
\usepackage{wrapfig}
\usepackage{multirow}
\usepackage{amsfonts}
\usepackage{cite}
\usepackage{amssymb}
\usepackage{caption}
\usepackage{algorithm}
\usepackage{algorithmic}
\usepackage{enumerate}
\usepackage{diagbox}
\usepackage{subfig}
\usepackage{wrapfig}
\usepackage{color}

\usepackage{times}
\usepackage{graphicx} 
\usepackage{algorithm}
\usepackage{algorithmic}
\usepackage{enumerate}
\usepackage{diagbox}
\usepackage{hyperref}

\usepackage[a4paper,margin=1in,footskip=0.25in]{geometry}

\makeatletter

\newcommand{\Rmnum}[1]{\expandafter\@slowromancap\romannumeral #1@}
\makeatother

\spnewtheorem{claim}{Claim}{\bfseries}{\rmfamily}

\begin{document}

\title{A Novel Geometric Approach for Outlier Recognition in High Dimension}
\author{Hu Ding and Mingquan Ye}
\institute{
 Department of Computer Science and Engineering \\
Michigan State University\\
  \email{\{huding,yemingqu\}@msu.edu}\\
}
\maketitle

\thispagestyle{empty}

\begin{abstract}
Outlier recognition is a fundamental problem in data analysis and has attracted a great deal of attention in the past decades. However, most existing methods still suffer from several issues such as high time and space complexities or unstable performances for different datasets. In this paper, we provide a novel algorithm for outlier recognition in high dimension based on the elegant geometric technique ``core-set". The algorithm needs only linear time and space complexities and achieves a solid theoretical quality guarantee. Another advantage over the existing methods is that our algorithm can be naturally extended to handle multi-class inliers. Our experimental results show that our algorithm outperforms existing algorithms on both random and benchmark datasets.
\end{abstract}
\section{Introduction}
In this big data era, we are confronted with an extremely large amount of data and it is important to develop efficient algorithmic techniques to handle the arising realistic issues. Due to its recent rapid development, deep learning~\cite{hinton2012improving} becomes a powerful tool for many emerging applications; meanwhile, the quality of training dataset often plays a key role and seriously affects the final learning result. For example, we can collect tremendous data (e.g., texts or images) through the internet, however, the obtained dataset often contains a significant amount of outliers. Since manually removing outliers will be very costly, it is very necessary to develop some efficient algorithms for recognizing outliers automatically in many scenarios.

{\em Outlier recognition} is a typical unsupervised learning problem and its counterpart in supervised learning is usually called {\em anomaly detection}~\cite{tan2006introduction}. In anomaly detection, the given training data are always positive and the task is to generate a model to depict the positive samples. Therefore, any new data can be distinguished to be positive or negative (i.e., anomaly) based on the obtained model. Several existing methods, especially for image data, include autoencoder~\cite{sakurada2014anomaly} and sparse coding~\cite{lu2013abnormal}.

Unlike anomaly detection, the given data for outlier recognition are unlabeled; thus we can only model it as an optimization problem based on some reasonable assumption in practice. For instance, it is very natural to assume that the inliers (i.e., normal data) locate in some dense region while the outliers are scattered in the feature space. Actually, many well known outlier recognition methods are based on this assumption~\cite{breunig2000lof,ester1996density}. However, most of the density-based methods are only for low-dimensional space and quite limited for large-scale high-dimensional data that are very common in computer vision problems (note that several high-dimensional approaches often are of heuristic natures and need strong assumptions~\cite{kriegel2009outlier,kriegel2008angle,aggarwal2001outlier}). Recently,~\cite{liu2014unsupervised} applied the one-class support vector machine (SVM) method~\cite{scholkopf1999support} to high-dimensional outlier recognition. Further,~\cite{xia2015learning} introduced a new unsupervised model of autoencoder inspired by the observation that inliers usually have smaller reconstruction errors than outliers.

\noindent\textbf{Our main contributions.} Although the aforementioned methods could efficiently solve the problem of outlier recognition to a certain extent, they still suffer from several issues such as high time and space complexities or unstable performances for different datasets. In this paper, we present a novel geometric approach for outlier recognition. Roughly speaking, we try to build an approximate minimum enclosing ball (MEB) to cover the inliers but exclude the outliers. This model is seemed to be very simple but involves a couple of computational challenges. For example, the existence of outliers makes the problem to be not only non-convex but also highly combinatorial. Also, the high dimensionality makes the problem more difficult. To tackle these challenges, we develop a randomized algorithmic framework using a popular geometric concept called ``core-set". Comparing with existing results for outlier recognition, we provide a thorough analysis on the complexities and quality guarantee. Moreover, we propose a simple greedy peeling strategy to extend our method to multi-class inliers. Finally, we test our algorithm on both random and benchmark datasets and the experimental results reveal the advantages of our approach over various existing methods.
\subsection{Other Related Work}
\label{sec-related}
Besides the aforementioned existing results, many other methods for outlier recognition/anomaly detection were developed previously and the readers can refer to several excellent surveys~\cite{kriegeloutlier,chandola2009anomaly,gupta2014outlier}.

In computational geometry, a {\em core-set}~\cite{agarwal2005geometric} is a small set of points that approximate the shape of a much larger point set, and thus can be used to significantly reduce the time complexities for many optimization problems (please refer to a recent survey~\cite{DBLP:journals/corr/Phillips16}). In particular, a core-set can be applied to efficiently compute an approximate MEB for a set of points in high-dimensional space~\cite{BHI,DBLP:journals/jea/KumarMY03}. Moreover,~\cite{badoiu2003smaller} showed that it is possible to find a core-set of size $\lceil 2/\epsilon\rceil$ that yields a $(1+\epsilon)$-approximate MEB, with an important advantage that the size is independent of the original size and dimensionality of the dataset. In fact, the algorithm for computing the core-set of MEB is a {\em Frank-Wolfe} style algorithm~\cite{frank1956algorithm}, which has been systematically studied by Clarkson~\cite{C10}.

The problem of MEB with outliers also falls under the umbrella of the topic {\em robust shape fitting}~\cite{har2004shape,agarwal2008robust}, but most of the approaches cannot be applied to high-dimensional data. ~\cite{zarrabistreaming} studied MEB with outliers in high dimension, however, the resulting approximation is a constant $2$ that is not fit enough for the applications proposed in this paper.

Actually, our idea is inspired by a recent work about removing outliers for SVM~\cite{ding2015random}, where they proposed a novel combinatorial approach called {\em Random Gradient Descent (RGD) Tree}. It is known that SVM is equivalent to finding the {\em polytope distance} from the origin to the {\em Minkowski Difference} of the given two labeled point sets. Gilbert algorithm~\cite{gilbert1966iterative,GJ09} is an efficient Frank-Wolfe algorithm for computing polytope distance, but a significant drawback is that the performance is too sensitive to outliers.  To remedy this issue, RGD Tree accommodates the idea of randomization to Gilbert algorithm. Namely, it selects a small random sample in each step by a carefully designed strategy to overcome the adverse effect from outliers.
\subsection{Preliminaries}
\label{sec-pre}
As mentioned before, we model outlier recognition as a problem of MEB with outliers in high dimension. Here we first introduce several definitions that are used throughout the paper.
\begin{definition}[Minimum Enclosing Ball (MEB)]
\label{def-meb}
Given a set $P$ of points in $\mathbb{R}^d$, MEB is the ball covering all the points with the smallest radius. The MEB is denoted by $MEB(P)$.
\end{definition}
\begin{definition}[MEB with Outliers]
\label{def-outlier}
Given a set $P$ of $n$ points in $\mathbb{R}^d$ and a small parameter $\gamma\in (0,1)$, MEB with outliers is to find the smallest ball that covers at least $(1-\gamma)n$ points. Namely, the task is to find a subset of $P$ having at least $(1-\gamma)n$ points such that the resulting MEB is the smallest among all the possible choices; the induced ball is denoted by $MEB(P, \gamma)$.
\end{definition}
From Definition~\ref{def-outlier} we can see that the major challenge is to determine the subset of $P$ which makes the problem a challenging combinatorial optimization. Therefore we relax our goal to its approximation as follows. For the sake of convenience, we always use $P_{\text{opt}}$ to denote the optimal subset of $P$, that is, $P_{\text{opt}}=\arg_{P'}\min\{$ the radius of $MEB(P')\mid P'\subset P, \left|P'\right|\geq (1-\gamma)n\}$, and $r_{\text{opt}}$ to denote the radius of $MEB(P_{\text{opt}})$.
\begin{definition}[Bi-criteria Approximation]
\label{def-app}
Given an instance $(P, \gamma)$ for MEB with outliers and two small parameters $0<\epsilon, \delta<1$, an $(\epsilon, \delta)$-approximation is a ball that covers at least $(1-(1+\delta)\gamma)n$ points and has the radius at most $(1+\epsilon)r_{\text{opt}}$.
\end{definition}
When both $\epsilon$ and $\delta$ are small, the bi-criteria approximation is very close to the optimal solution with only a slight violation on the number of covering points and radius.

The rest of the paper is organized as follows. We introduce our main algorithm and the theoretical analyses in Section~\ref{sec-tech}. The experimental results are shown in Section~\ref{sec-exp}. Finally, we extend our method to handle multi-class inliers in Section~\ref{sec-extension}.
\section{Our Algorithm and Analyses}
\label{sec-tech}
In this section, we present our method in detail. For the sake of completeness, we first briefly introduce the core-set for MEB based on the idea of~\cite{badoiu2003smaller}.

The algorithm is a simple iterative procedure with an elegant analysis: initially, it selects an arbitrary point and places it into a set $S$ that is empty at the beginning; in each of the following $\lceil{2/\epsilon}\rceil$ steps, the algorithm updates the center of $MEB(S)$ and adds the farthest point to $S$; finally, the center of $MEB(S)$ induces a $(1+\epsilon)$-approximation for MEB of the whole input point set. The selected $\lceil{2/\epsilon}\rceil$ points are also called the core-set for MEB. To ensure there is at least certain extent of improvement achieved in each iteration,~\cite{badoiu2003smaller} showed that the following two inequalities would hold if the algorithm always selects the farthest point to the temporary center of $MEB(S)$:
\begin{align}
r_{i+1}&\geq (1+\epsilon)r_{\text{opt}}-L_i, \\
r_{i+1}&\geq \sqrt{r^2_i+L^2_i},
\label{for-meb}
\end{align}
where $r_i$ and $r_{i+1}$ are the radii of the $i$-th and the $(i+1)$-th iterations respectively, $r_{\text{opt}}$ is the optimal radius of the MEB, and $L_i$ is the shifting distance of the center of $MEB(S)$.
\begin{algorithm}
   \caption{$(\epsilon,\delta)$-approximation Algorithm of Outlier Recognition Problem}
   \label{alg1}
\begin{algorithmic}[1]
   \INPUT A point set $P$ with $n$ points in $\mathbb{R}^{d}$, the fraction of outliers $\gamma\in(0,1)$ and four parameters $0<\epsilon,\delta,\mu<1,h\in\mathbb{Z}^{+}$.
   \OUTPUT A tree with each node whose attached point is a candidate for the $(\epsilon,\delta)$-approximation solutions.
\STATE Each node $v$ in the tree is associated with a point (with a slight abuse of notation, we also use $v$ to denote the point). Initially, randomly pick a point from $P$ as root node $r$.
\STATE Starting with root, grow each node as follows:
\begin{enumerate}[(1)]
\item Let $v$ be the current node.
\item If the height of $v$ is $h$, $v$ becomes a leaf node. Otherwise, perform the following steps:
\begin{enumerate}[(a)]
\item Let $\mathcal{P}^r_v$ denote the set of points along the path from root $r$ to node $v$, and $c_v$ denote the center of $MEB(\mathcal{P}^r_v)$. We say that $c_{v}$ is the attached point of $v$.
\item Let $k=(1+\delta)\gamma n$. Compute the point set $P_{v}$ containing the top $k$ points which have the largest distances to $c_v$.
\item Take a random sample $S_{v}$ of size $(1+\frac{1}{\delta})\ln\frac{h}{\mu}$ from $P_{v}$, and let each point $v'\in S_{v}$ be a child node of $v$.
\end{enumerate}
\end{enumerate}
\end{algorithmic}
\end{algorithm}
\vspace{-0.1cm}
\subsection{Algorithm for MEB with Outliers}
\label{sec-alg}
We present our $(\epsilon, \delta)$-approximation algorithm for MEB with outliers in this section. Although the outlier recognition problem belongs to unsupervised learning, we can estimate the fraction of outliers in the given data before executing our algorithm. In practice, we can randomly collect a small set of samples from the given data, and manually identify the outliers and estimate the outlier ratio $\gamma$. Therefore, in this paper we assume that the outlier ratio is known.

To better understand our algorithm, we first illustrate the high-level idea.  If taking a more careful analysis on the previously mentioned core-set construction algorithm~\cite{badoiu2003smaller}, we can find that it is not necessary to select the farthest point to the center of $MEB(S)$ in each step. Instead, as long as the selected point has a distance larger than $(1+\epsilon)r_{\text{opt}}$, the minimal extent of improvement would always be guaranteed~\cite{Din11}. As a consequence, we investigate the following approach.

We denote the ball centered at point $c$ with radius $r>0$ as $Ball(c, r)$. Recall that $P_{\text{opt}}$ is the subset of $P$ yielding the optimal MEB with outliers, and $r_{\text{opt}}$ is the radius of $MEB(P_{\text{opt}})$ (see Section~\ref{sec-pre}). In the $i$-th step, we add an arbitrary point from $P_{opt}\setminus Ball(c_i, (1+\epsilon)r_{\text{opt}})$ to $S$ where $c_i$ is the current center of $S$. Based on the above observation, we know that a $(1+\epsilon)$-approximation is obtained after at most $\lceil{2/\epsilon}\rceil$ steps, that is, $\left|P\cap Ball(c_i, (1+\epsilon)r_{\text{opt}})\right|\geq (1-\gamma)n$ when $i\geq \lceil{2/\epsilon}\rceil$.

However, in order to carry out the above approach we need to solve two key issues: how to determine the value of $r_{\text{opt}}$ and how to select a point belonging to $P_{opt}\setminus Ball(c_i, (1+\epsilon)r_{\text{opt}})$. Actually, we can implicitly avoid the first issue via replacing the radius $(1+\epsilon)r_{\text{opt}}$ by the $k$-th largest distance from the points of $P$ to $c_i$, where $k$ is some appropriate number that will be determined in our following analysis. For the second issue, we have to take a small random sample instead of a single point from $P_{opt}\setminus Ball(c_i, (1+\epsilon)r_{opt})$ and try each of the sampled points; this operation will result in a tree structure that is similar to the RGD Tree introduced by~\cite{ding2015random} for SVM. We present the algorithm in Algorithm~\ref{alg1} and place the detailed parameter settings, proof of correctness, and complexity analyses in Sections~\ref{sec-quality} \&~\ref{sec-complexity}.
\begin{figure}[h]
\centering
\includegraphics[scale=0.35]{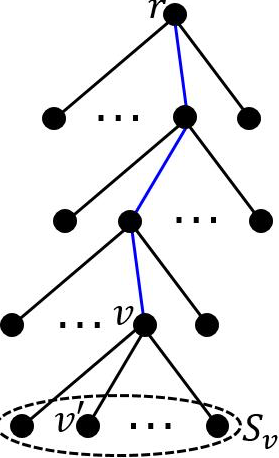}
\caption{The blue links represent the path from root $r$ to node $v$, and $\mathcal{P}^r_v$ contains the four points along the path. The point set $S_v$ corresponds to the child nodes of $v$.}
\label{fig2}
\end{figure}
\vspace{-3ex}

We illustrate Step 2(2)(a-c) of Algorithm~\ref{alg1} in Fig.~\ref{fig2}.
\subsection{Parameter Settings and Quality Guarantee}
\label{sec-quality}
We denote the tree constructed by Algorithm~\ref{alg1} as $\mathbb{H}$. The following theorem shows the success probability of Algorithm~\ref{alg1}.
\begin{theorem}
\label{theorem1}
If we set $h=\lceil{\frac{2}{\epsilon}}\rceil+1$, then with probability at least $(1-\mu)(1-\gamma)$ there exists at least one node of $\mathbb{H}$ yielding  an $(\epsilon,\delta)$-approximation for the problem of MEB with outliers.
\end{theorem}
Before proving Theorem~\ref{theorem1}, we need to introduce several important lemmas.
\begin{lemma}~\cite{DX14}
\label{lemma1}
Let $Q$ be a set of elements, and $Q'$ be a subset of $Q$ with
size $\left|Q'\right|=\beta\left|Q\right|$ for some $\beta\in(0,1)$. If one randomly samples $\frac{1}{\beta}\ln\frac{1}{\eta}$ elements from $Q$, then with probability at least $1-\eta$, the sample contains at least one element in $Q'$ for any $0<\eta<1$.
\end{lemma}
\begin{lemma}
For each node $v$, the set $S_{v}$ in Algorithm~\ref{alg1} contains at least one point from $P_{\text{opt}}$ with probability $1-\frac{\mu}{h}$.
\label{lemma2}
\end{lemma}
\begin{proof}
Since $\left|P_{v}\right|=(1+\delta)n\gamma$ and $\left|P\backslash P_{\text{opt}}\right|=n\gamma$, we have
\begin{equation}
\begin{aligned}
\frac{\left|{P_{v}\cap P_{\text{opt}}}\right|}{\left|P_{v}\right|}&=1-\frac{\left|{P_{v}\backslash P_{\text{opt}}}\right|}{\left|P_{v}\right|}\\
&\ge1-\frac{\left|{P\backslash P_{\text{opt}}}\right|}{\left|P_{v}\right|}=\frac{\delta}{1+\delta}.
\end{aligned}
\end{equation}
Note that the size of $S_{v}$ is $(1+\frac{1}{\delta})\ln\frac{h}{\mu}$. If we apply Lemma~\ref{lemma1} via setting $\beta=\frac{\delta}{1+\delta}$ and $\eta=\frac{\mu}{h}$, it is easy to know that $S_{v}$ contains at least one point from $P_{\text{opt}}$ with probability $1-\frac{\mu}{h}$.
\end{proof}
\begin{lemma}
With probability $(1-\gamma)(1-\mu)$, there exists a leaf node $u\in\mathbb{H}$ such that the corresponding set $\mathcal{P}^r_u\subset P_{\text{opt}}$.
\label{lemma3}
\end{lemma}
\begin{proof}
Lemma~\ref{lemma2} indicates that each node $v$ has a child node corresponding to a point from $P_{\text{opt}}$ with probability $1-\frac{\mu}{h}$. In addition, the probability of root $r$ belonging to $P_{\text{opt}}$ is $1-\gamma$ (recall that $\gamma$ is the fraction of outliers). Note that the height of $\mathbb{H}$ is $h$, then with probability at least
\begin{equation}
(1-\gamma)\left({1-\frac{\mu}{h}}\right)^{h}>(1-\gamma)(1-\mu),
\end{equation}
there exists one leaf node $u\in\mathbb{H}$ satisfying $\mathcal{P}^r_u\subset P_{\text{opt}}$.
\end{proof}

In the remaining analyses, we always assume that such a root-to-leaf path $\mathcal{P}^r_u$ described in Lemma~\ref{lemma3} exists and only focus on the nodes along this path. We denote $\hat{R}=(1+\epsilon)r_{\text{opt}}$ where $r_{\text{opt}}$ is the optimal radius of the MEB with outliers. Let $Ball(c_v, r_v)$ be the MEB covering $P\setminus P_v$ centered at $c_v$, and the radii of $MEB(\mathcal{P}^r_v)$ and $MEB(\mathcal{P}^r_{v'})$ be $\tilde{r}_{v}$ and $\tilde{r}_{v'}$ respectively. Readers can refer to Fig.~\ref{fig3}.
The following lemma is a key observation for MEB.

\begin{lemma}~\cite{badoiu2003smaller}
Given a set $P$ of points in $\mathbb{R}^d$, let $r_P$ and $c_P$ be the radius and center of $MEB(P)$ respectively. Then for any point $p\in\mathbb{R}^d$ with a distance $K\geq 0$ to $c_P$, there exists a point $q\in P$ such that $\left\|p-q\right\|\geq \sqrt{r^2_P+K^2}$.
\label{lemma4}
\end{lemma}
The following lemma is a key for proving the quality guarantee of Algorithm~\ref{alg1}. As mentioned in Section~\ref{sec-alg}, the main idea follows the previous works~\cite{badoiu2003smaller,Din11}. For the sake of completeness, we present the detailed proof here.
\begin{lemma}
For each node $v\in\mathcal{P}^r_u$, at least one of the following two events happens: (1) $c_v$ is an $(\epsilon,\delta)$-approximation; (2) its child $v'$ on the path $\mathcal{P}^r_u$ satisfies
\begin{equation}
\tilde{r}_{v'}\ge\frac{\hat{R}}{2}+\frac{\tilde{r}_{v}^{2}}{2\hat{R}}.\label{for-lem5}
\end{equation}
\label{lemma5}
\end{lemma}
\vspace{-5ex}
\begin{proof}
If $r_{v}\le\hat{R}$, then we are done; that is, $Ball(c_v, r_v)$ covers $(1-(1+\delta)\gamma)n$ points and $r_v\leq (1+\epsilon)r_{\text{opt}}$. Otherwise, $r_{v}>\hat{R}$ and we consider the second case.

By triangle inequality and the fact that $v'$ (i.e., the point associating the node ``$v'$") lies outside $Ball(c_v, r_v)$, we have
\begin{equation}
\left\|{c_v-c_{v'}}\right\|+\left\|{c_{v'}-v'}\right\|\ge\left\|{c_v-v'}\right\|> r_{v}>\hat{R}.
\end{equation}
Let $\left\|c_{v}-c_{v'}\right\|=K_{v}$. Combining the fact that $\left\|c_{v'}-v'\right\|\le\tilde{r}_{v'}$, we have
\begin{equation}
\tilde{r}_{v'}>\hat{R}-K_{v}.
\label{eq1}
\end{equation}
\vspace{-2ex}
\begin{figure}[h]
\centering
\includegraphics[scale=0.35]{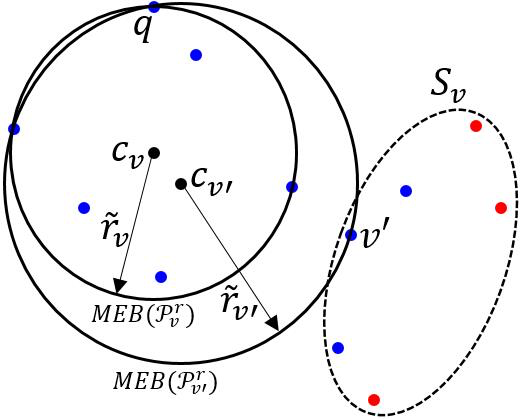}
\caption{The illustration of $MEB(\mathcal{P}^{r}_{v})$ and $MEB(\mathcal{P}^{r}_{v'})$; the blue and red points represent the inliers and outliers, respectively.}
\label{fig3}
\end{figure}
\vspace{-2ex}

By Lemma~\ref{lemma4}, we know that there exists one point $q$ (see Fig.~\ref{fig3}) in $MEB(\mathcal{P}^r_v)$ satisfying $\left\|q-c_{v'}\right\|\ge\sqrt{\tilde{r}_{v}^{2}+K_{v}^{2}}$. Since $q$ is also inside $MEB(\mathcal{P}^r_{v'})$, $\left\|q-c_{v'}\right\|\le\tilde{r}_{v'}$. Then, we have
\begin{equation}
\tilde{r}_{v'}\ge\sqrt{\tilde{r}_{v}^{2}+K_{v}^{2}}.
\label{eq2}
\end{equation}
Combining (\ref{eq1}) and (\ref{eq2}), we obtain
\begin{equation}
\tilde{r}_{v'}\ge\max\left\{\hat{R}-K_{v},\sqrt{\tilde{r}_{v}^{2}+K_{v}^{2}}\right\}. \label{for-max}
\end{equation}
Because $\hat{R}-K_{v}$ and $\sqrt{\tilde{r}_{v}^{2}+K_{v}^{2}}$ are decreasing and increasing on $K_v$ respectively, we let $\hat{R}-K_{v}=\sqrt{\tilde{r}_{v}^{2}+K_{v}^{2}}$ to achieve the lower bound (i.e., $K_{v}=\frac{\hat{R}^{2}-\tilde{r}_{v}^{2}}{2\hat{R}}$). Substituting the value of $K_{v}$ to (\ref{for-max}), we have $\tilde{r}_{v'}\ge\frac{\hat{R}}{2}+\frac{\tilde{r}_{v}^{2}}{2\hat{R}}$. As a consequence, the second event happens and the proof is completed.
\end{proof}

Now we prove Theorem~\ref{theorem1} by the idea from~\cite{badoiu2003smaller}. Suppose no node in $\mathcal{P}^{r}_{u}$ makes the first event of Lemma~\ref{lemma5} occur. As a consequence, we obtain a series of inequalities for each pair of radii $\tilde{r}_{v'}$ and $\tilde{r}_{v}$ (see (\ref{for-lem5})). For ease of analysis, we denote $\tilde{r}_v=\lambda_{i}\hat{R}$ if the height of $v$ is $i$ in $\mathbb{H}$. By Inequality (\ref{for-lem5}), we have
\begin{equation}
\label{eq3}
\lambda_{i+1}\ge\frac{1+\lambda_{i}^{2}}{2}.
\end{equation}
Combining the initial case $\lambda_{1}=0$ and (\ref{eq3}), we obtain
\begin{equation}
\label{eq4}
\lambda_{i}\ge 1-\frac{2}{i+1}
\end{equation}
by induction~\cite{badoiu2003smaller}. Note that the equality in (\ref{eq4}) holds only when $i=1$, therefore,
\begin{equation}
\lambda_{h}>1-\frac{2}{h+1}=1-\frac{2}{\lceil{\frac{2}{\epsilon}}\rceil+2}\ge 1-\frac{2}{\frac{2}{\epsilon}+2}=\frac{1}{1+\epsilon}.
\end{equation}
Then, $\tilde{r}_{u}=\lambda_{h}\hat{R}>r_{\text{opt}}$ (recall that $u$ is the leaf node on the path $\mathcal{P}^{r}_{u}$), which is a contradiction to our assumption $\mathcal{P}^{r}_{u}\subset P_{\text{opt}}$. The success probability directly comes from Lemma~\ref{lemma3}. Overall, we obtain Theorem~\ref{theorem1}.
\subsection{Complexity Analyses}
\label{sec-complexity}
We analyze the time and space complexities of Algorithm~\ref{alg1} in this section.

\textbf{Time Complexity}. For each node $v$, we need to compute the corresponding approximate $MEB(\mathcal{P}^r_v)$. To avoid computing the exact MEB costly, we apply the approximation algorithm proposed by~\cite{badoiu2003smaller}. See Algorithm~\ref{alg2} for details.
\begin{algorithm}[H]
   \caption{Approximation Algorithm of MEB}
   \label{alg2}
\begin{algorithmic}[1]
   \INPUT A point set $Q$ in $\mathbb{R}^{d}$, and $N\in Z^+$.
   \STATE Start with an arbitrary point $c_{1}\in Q$, $t\leftarrow 1$.
   \WHILE{$t<N$}
       \STATE Find the point $q\in Q$ farthest away from $c_{t}$.
       \STATE $c_{t+1}\leftarrow c_{t}+\frac{1}{t+1}(q-c_{t})$.
       \STATE $t\leftarrow t+1$.
   \ENDWHILE
   \STATE \textbf{return} $c_{t}$.
\end{algorithmic}
\end{algorithm}
\vspace{-0.4cm}
For Algorithm~\ref{alg2}, we have the following theorem.
\begin{theorem}
\label{the-alg2}
~\cite{badoiu2003smaller} Let the center and radius of $MEB(Q)$ be $c_Q$ and $r_Q$ respectively, then $\forall t$, $\left\|{c_{Q}-c_{t}}\right\|\le \frac{r_{Q}}{\sqrt{t}}$.
\end{theorem}
From Theorem~\ref{the-alg2}, we know that a $(1+\varepsilon)$-approximation for MEB can be obtained when $N=1/\varepsilon^2$ with the time complexity $O\left(\frac{\left|{Q}\right|d}{\varepsilon^2}\right)$. Suppose the height of node $v$ is $i$, then the complexity for computing the corresponding approximate $MEB(\mathcal{P}^r_v)$ is $O\left(\frac{id}{\varepsilon^{2}}\right)$. Further, in order to obtain the point set $P_{v}$, we need to find the pivot point that has the $(n-k)$-th smallest  distance to $c_v$. Here we apply the PICK algorithm~\cite{blum1973time} which can find the $l$-th smallest from a set of $n$ ($l\le n$) numbers in linear time. Consequently, the complexity for each node $v$ at the $i$-th layer is $O\left({\left({n+\frac{i}{\varepsilon^{2}}}\right)d}\right)$. Recall that there are  $\left|{S_{v}}\right|^{i-1}$ nodes at the $i$-th layer of $\mathbb{H}$. In total, the time complexity of our algorithm is
\begin{equation}
T=\sum_{i=1}^{h}\left({\left({1+\frac{1}{\delta}}\right)\ln\frac{h}{\mu}}\right)^{i-1}\left({n+\frac{i}{\varepsilon^2}}\right)d.
\end{equation}
If we assume $1/\varepsilon$ is a constant, the complexity $T=O(Cnd)$ is linear in $n$ and $d$, where the hidden constant $C=\left({\left({1+\frac{1}{\delta}}\right)\ln\frac{h}{\mu}}\right)^{h-1}$. In our experiment, we can carefully choose the parameters $\delta,\epsilon,\mu$ so as to keep the value of $C$ not too large.

\textbf{Space Complexity}. In our implementation, we use a queue $\mathcal{Q}$ to store the nodes in the tree. When the head of $\mathcal{Q}$ is popped, its $\left|S_{v}\right|$ child nodes are pushed into $\mathcal{Q}$. In other words, we simulate {\em breadth first search} on the tree $\mathbb{H}$. Therefore, $\mathcal{Q}$ always keeps its size at most $C=\left({\left({1+\frac{1}{\delta}}\right)\ln\frac{h}{\mu}}\right)^{h-1}$. Note that each node $v$ needs to store $\mathcal{P}^r_v$ to compute its corresponding MEB, but actually we only need to record the pointers to link the points in $\mathcal{P}^{r}_{v}$. Therefore, the space complexity of $\mathcal{Q}$ is $O(Ch)$. Together with the space complexity of the input data, the total space complexity of our algorithm is $O(Ch+nd)$.
\subsection{Boosting}
\label{sec3.3}
By Theorem~\ref{theorem1}, we know that with probability at least $(1-\mu)(1-\gamma)$ there exists an $(\epsilon,\delta)$-approximation in the resulting tree. However, when outlier ratio is high, say $\gamma=0.5$, the success probability $(1-\gamma)(1-\mu)$ will become small. To further improve the performance of our algorithm, we introduce the following two boosting methods.
\begin{enumerate}
\item \textbf{Constructing a forest.} Instead of building a single tree, we randomly initialize several root nodes and grow each root node to be a tree. Suppose the number of root nodes is $\kappa$. The probability that there exists an $(\epsilon,\delta)$-approximation in the forest is at least $1-(1-(1-\gamma)(1-\mu))^{\kappa}$ which is much larger than  $(1-\gamma)(1-\mu)$.
\item \textbf{Sequentialization.} First, initialize one root node and build a tree. Then select the best node in the tree and set it to be the root node for the next tree. After iteratively performing the procedure for several rounds, we can obtain a much more robust solution.
\end{enumerate}
\section{Experiments}
\label{sec-exp}

From our analysis in Section~\ref{sec-quality}, we know that  Algorithm~\ref{alg1} results in a tree $\mathbb{H}$ where each node $v$ has a candidate $c_v$ for the desired $(\epsilon, \delta)$-approximation for the problem of MEB with outliers. For each candidate, we identify the nearest $(1-(1+\delta)\gamma)n$ points to $c_v$ as the inliers. To determine the final solution, we select the candidate that has the smallest variance of the inliers.
\subsection{Datasets and Methods to Be Compared}
In our experiment, we test the algorithms on two random datasets and two benchmark image datasets.
In terms of the random datasets, we  generate the data points based on normal and uniform distributions under the assumption that the inliers usually locate in dense regions while the outliers are scattered in the space.
The benchmark image datasets include the popular MNIST~\cite{lecun1998gradient} and Caltech~\cite{fei2007learning}.

To make our experiment more convincing, we compare our algorithm with three well known methods for outlier recognition: angle-based outlier detection (ABOD)~\cite{kriegel2008angle}, one-class SVM (OCSVM)~\cite{scholkopf1999support}, and discriminative reconstructions in an autoencoder (DRAE)~\cite{xia2015learning}. Specifically, ABOD distinguishes the inliers and outliers by assessing the distribution of the angles determined by each 3-tuple data points; OCSVM models the problem of outlier recognition as a soft-margin one-class SVM; DRAE applies autoencoder to separate the inliers and outliers based on their reconstruction errors.


The performances of the algorithms are measured by the commonly used {\em $F1$ score} $=\frac{2*\text{Precision}*\text{Recall}}{\text{Precision}+\text{Recall}}$,
where precision is the proportion of the correctly identified positives relative to the total number of identified positives, and recall is the proportion of the correctly identified  positives relative to the total number of positives in the dataset.

\subsection{Random Datasets}
\label{sec4.2}
We validate our algorithm on the following two random datasets.

\textbf{A toy example in $2$D.} To better illustrate the intuition of our algorithm, we first run it on a random dataset in 2D. We generate an instance of $10,000$ points with the outlier ratio $\gamma=0.4$. The inliers are generated by a normal distribution; the outliers consist of four groups where the first three are generated by normal distributions and the last is generated by a uniform distribution. The four groups of outliers contains $800$, $1200$, $800$, and $1200$ points, respectively.
See Fig.~\ref{fig4}. The red circle obtained by our algorithm is the boundary to distinguish the inliers and outliers where the resulting $F1$ score is $0.944$. From this case, we can see that our algorithm can efficiently recognize the densest region even if the outlier ratio is high and the outliers also form some dense regions in the space.

\vspace{-2ex}
\begin{figure}[h]
\begin{center}
\centerline{\includegraphics[scale=0.4]{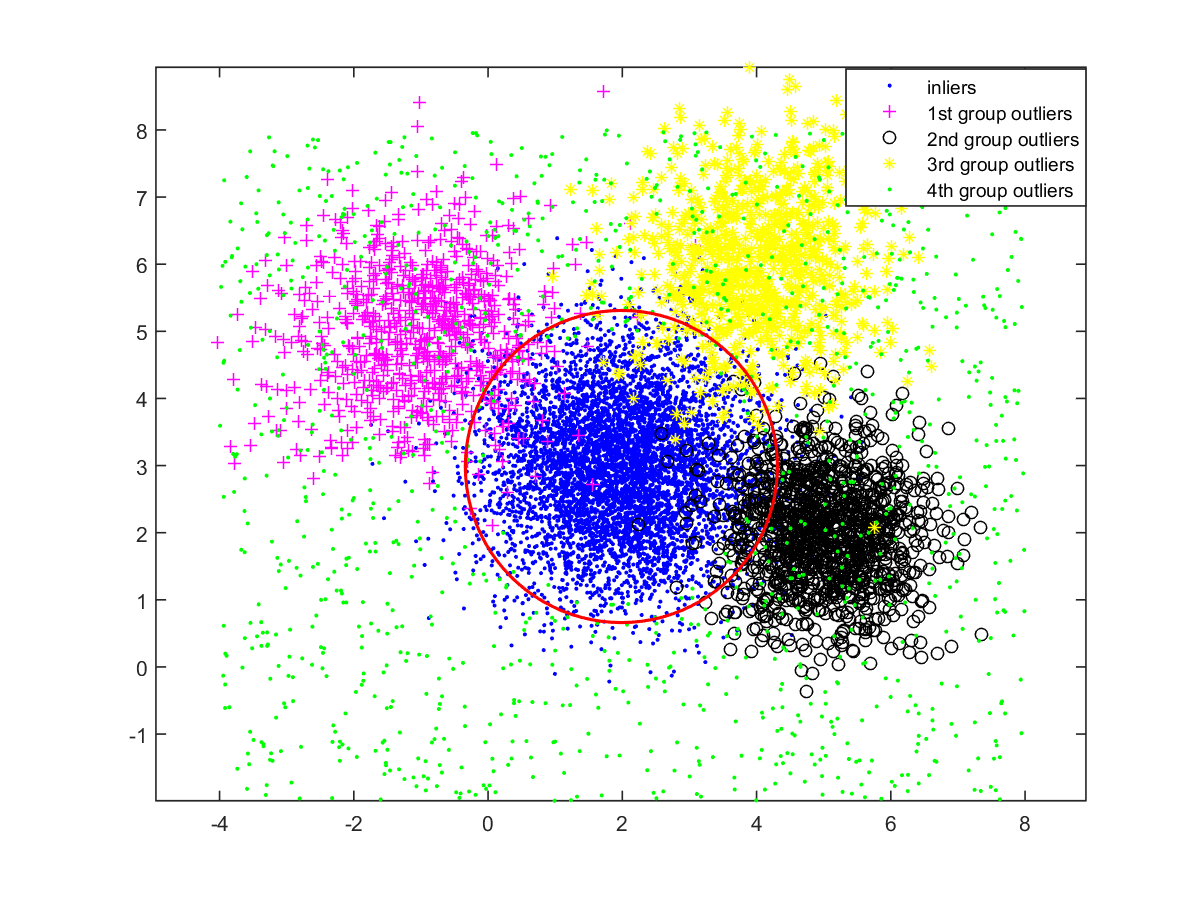}}
\vspace{-0.7cm}
\caption{The illustration of our algorithm on a $2$-dimensional point set.}
\label{fig4}
\end{center}
\end{figure}
\vspace{-0.6cm}

\textbf{High-Dimensional Points.} We further test our algorithm and the other three methods on high-dimensional dataset. Similar to the previous 2D case, we generate $20,000$ points with four groups of outliers in $\mathbb{R}^{100}$; the outlier ratio $\gamma$ varies from $0.1$ to $0.5$.
%
The $F1$ scores are displayed in Table~\ref{tab1}, from which we can see that our algorithm significantly outperforms the other three methods for all the levels of outlier ratio.

\vspace{-1ex}
\begin{table}[htbp]
\centering
\caption{The $F1$ scores for the high-dimensional random dataset.}
\vspace{2ex}
\begin{tabular}{|l|c|c|c|c|c|}
\hline
\diagbox[height=1cm,width=2.2cm]{Methods}{$\gamma$} &$0.1$ &$0.2$ &$0.3$ &$0.4$ &$0.5$\\
\hline
ABOD &$0.907$ &$0.815$ &$0.705$ &$0.586$ &$0.419$\\
\hline
OCSVM &$0.967$ &$0.926$ &$0.880$ &$0.827$ &$0.745$\\
\hline
DRAE &$0.951$ &$0.889$ &$0.809$ &$0.709$ &$0.572$\\
\hline
Ours &$\textbf{0.984}$ &$\textbf{0.965}$ &$\textbf{0.939}$ &$\textbf{0.938}$ &$\textbf{0.898}$\\
\hline
\end{tabular}
\label{tab1}
\end{table}
\vspace{-1ex}
\subsection{Benchmark Image Datasets}
\label{sec4.3}
In this section, we evaluate all the four methods on two benchmark image datasets.

\subsubsection{MNIST Dataset}
\label{sec4.3.1}
MNIST contains $70,000$ handwritten digits ($0$ to $9$) composed of both training and test datasets. For each of the 10 digits, we add the outliers by randomly selecting the images  from the other $9$ digits. For each outlier ratio $\gamma$, we compute the average $F1$ score over all the 10 digits. To map the images to a feature (Euclidean) space, we use two kinds of image features: PCA-grayscale and autoencoder feature.

\begin{table*}[htbp]
\centering
\caption{The $F1$ scores of the four methods on MNIST by using PCA-grayscale; the three columns for each $\gamma$ correspond to PCA-$0.95$, PCA-$0.5$, and PCA-$0.1$, respectively.}
\vspace{2ex}
\small
\begin{tabular}{|l|c|c|c|c|c|c|}
\hline
\diagbox[height=0.8cm,width=2cm]{Methods}{$\gamma$} &$0.1$ &$0.2$ &$0.3$ &$0.4$ &$0.5$\\
\hline
ABOD  &$0.898,0.895,0.892$ &$0.775,0.774,0.771$ &$0.648,0.617,0.642$ &$0.500,0.470,0.496$ &$0.346,0.329,0.364$ \\
\hline
OCSVM &$0.937,\textbf{0.941},0.934$ &$0.874,0.883,0.867$ &$0.804,0.817,0.798$ &$0.725,0.740,0.713$ &$\textbf{0.648},0.639,0.605$ \\
\hline
DRAE &$0.913,0.908,0.911$ &$0.822,0.818,0.816$ &$0.726,0.722,0.711$ &$0.620,0.617,0.602$ &$0.531,0.501,0.488$ \\
\hline
Ours &$\textbf{0.939},\textbf{0.941},\textbf{0.936}$ &$\textbf{0.881},\textbf{0.891},\textbf{0.880}$ &$\textbf{0.822},\textbf{0.853},\textbf{0.823}$ &$\textbf{0.760},\textbf{0.778},\textbf{0.773}$ &$0.633,\textbf{0.658},\textbf{0.651}$ \\
\hline
\end{tabular}
\label{tab3}
\end{table*}
\normalsize

\vspace{-0.3cm}
\begin{enumerate}[(1)]
\item \textbf{PCA-grayscale Feature.} Each image in MNIST has a $28\times 28$ grayscale which is represented by a $784$-dimensional vector. Note that the images of MNIST have massive redundancy. For example, the digits often locate in the middle of the images and all the background pixels have the value of $0$. Therefore,
%
%
%
we apply principle component analysis (PCA) to reduce the redundancy by trying multiple projection matrices which preserve $95\%$, $50\%$, and $10\%$ energy of the original grayscale features. These three features are denoted as PCA-$0.95$, PCA-$0.5$ and PCA-$0.1$, respectively. The results are shown in Table~\ref{tab3}. We notice that our $F1$ scores always achieve the highest by PCA-$0.5$; this is due to the fact that PCA-$0.5$ can significantly reduce the redundancy as well as preserve the most useful information (comparing with PCA-$0.95$ and PCA-$0.1$).

\item \textbf{Autoencoder Feature.} Autoencoder~\cite{rumelhart1988learning} is often adopted to extract the features of grayscale images. The autoencoder model trained in our experiment has seven symmetrical hidden layers ($1000$-$500$-$250$-$60$-$250$-$500$-$1000$), and the input layer is a $784$-dimensional grayscale. We use the middle hidden layer as image feature. The results are shown in Table~\ref{tab4} and our method achieves the best for most of the cases.
\begin{table}[htbp]
\centering
\caption{The $F1$ scores of the four methods on MNIST by using autoencoder feature.}
\vspace{2ex}
\begin{tabular}{|l|c|c|c|c|c|c|}
\hline
\diagbox[height=1cm,width=2.2cm]{Methods}{$\gamma$} &$0.1$ &$0.2$ &$0.3$ &$0.4$ &$0.5$\\
\hline
ABOD  &$0.894$ &$0.778$ &$0.637$ &$0.479$ &$0.313$ \\
\hline
OCSVM &$0.906$ &$0.807$ &$0.706$ &$0.598$ &$0.496$ \\
\hline
DRAE &$\textbf{0.933}$ &$0.883$ &$0.819$ &$0.737$ &$0.625$ \\
\hline
Ours &$0.932$ &$\textbf{0.885}$ &$\textbf{0.831}$ &$\textbf{0.770}$ &$\textbf{0.694}$ \\
\hline
\end{tabular}
\label{tab4}
\end{table}
\end{enumerate}
\vspace{-0.6cm}
\subsubsection{Caltech Dataset}
\label{sec3.3.2}
The Caltech-$256$ dataset \footnote{\url{http://www.vision.caltech.edu/Image_Datasets/Caltech256/}} includes $256$ image sets.  We choose $11$ concepts as the inliers in our experiment, which are \textit{airplane, binocular, bonsai, cup, face, ketch, laptop, motorbike, sneaker, t-shirt}, and \textit{watch}. We apply \textbf{VGG net}~\cite{simonyan2014very} to extract the image features, which is the $4096$-dimensional output of the second fully-connected layer. The results are shown in Table~\ref{tab5}.
\vspace{-0.5cm}
\begin{table}[h]
\centering
\caption{The $F1$ scores of the four methods on Caltech-$256$ by using VGG net feature.}
\vspace{2ex}
\begin{tabular}{|l|c|c|c|c|c|c|}
\hline
\diagbox[height=1cm,width=2.2cm]{Methods}{$\gamma$} &$0.1$ &$0.2$ &$0.3$ &$0.4$ &$0.5$\\
\hline
ABOD &$0.945$ &$0.838$ &$0.707$ &$0.499$ &$0.233$ \\
\hline
OCSVM &$0.930$ &$0.885$ &$0.839$ &$0.783$ &$0.739$ \\
\hline
DRAE &$0.955$ &$0.937$ &$0.930$ &$\textbf{0.927}$ &$\textbf{0.912}$ \\
\hline
Ours &$\textbf{0.964}$ &$\textbf{0.948}$ &$\textbf{0.932}$ &$0.924$ &$0.906$ \\
\hline
\end{tabular}
\label{tab5}
\end{table}
\vspace{-0.2cm}

Unlike the random data, the distribution of real data in the space is much more complicated. To alleviate this problem, we try to capture the separate parts of the original VGG net feature. Similar to Section~\ref{sec4.3.1}, we apply PCA to reduce the redundancy of VGG net feature and preserve its key parts. Three matrices are obtained to preserve $95\%$, $50\%$, and $10\%$ energy respectively.
The results are shown in Table~\ref{tab9}. We can see that our method achieves the best for all the cases, especially when using PCA-$0.5$ (marked by underlines). More importantly,  PCA-$0.5$ considerably improves the results by using the original VGG net feature (see Table~\ref{tab5}), and the dimensionality is only $50$ which results in a significant reduction on the complexities.

\begin{table*}[htbp]
\centering
\caption{The $F1$ scores of the four methods on Caltech-$256$ by using PCA-VGG feature; the three columns for each $\gamma$ correspond to PCA-$0.95$, PCA-$0.5$, and PCA-$0.1$, respectively.}
\vspace{2ex}
\small
\begin{tabular}{|l|c|c|c|c|c|c|}
\hline
\diagbox[height=0.8cm,width=2cm]{Methods}{$\gamma$} &$0.1$ &$0.2$ &$0.3$ &$0.4$ &$0.5$\\
\hline
ABOD  &$0.944,0.942,0.941$ &$0.837,0.832,0.869$ &$0.707,0.708,0.715$ &$0.497,0.489,0.525$ &$0.223,0.199,0.288$ \\
\hline
OCSVM &$0.932,0.914,0.921$ &$0.884,0.894,0.867$ &$0.837,0.869,0.827$ &$0.782,0.830,0.771$ &$0.717,0.790,0.699$ \\
\hline
DRAE &$0.955,0.947,0.928$ &$0.918,0.924,0.878$ &$0.873,0.914,0.835$ &$0.873,0.902,0.773$ &$0.869,0.887,0.692$ \\
\hline
Ours &$\textbf{0.966},\underline{\textbf{0.986}},\textbf{0.949}$ &$\textbf{0.950},\underline{\textbf{0.984}},\textbf{0.923}$ &$\textbf{0.934},\underline{\textbf{0.978}},\textbf{0.897}$ &$\textbf{0.916},\underline{\textbf{0.973}},\textbf{0.871}$ &$\textbf{0.899},\underline{\textbf{0.958}},\textbf{0.844}$ \\
\hline
\end{tabular}
\label{tab9}
\end{table*}
\normalsize
\subsection{Comparisons of Time Complexities}
From Section~\ref{sec4.2} and Section~\ref{sec4.3} we know that our method achieves the robust and competitive performances in terms of accuracy.
In this section, we compare the time complexities of all the four algorithms.

\textbf{ABOD} has the time complexity $O(n^{3}d)$. In the experiment, we use its speed-up edition FastABOD which has the reduced time complexity $O((n^{2}+nk^{2})d)$ where $k$ is some specified parameter.

\textbf{OCSVM} is formulated as a quadratic programming with the time complexity $O(n^{3})$.

\textbf{DRAE} alternatively executes the following two steps: discriminative labeling and reconstruction learning.
 Suppose it runs in $N_{1}$ rounds; actually the two inner steps are also iterative procedures which both run $N_2$ iterations. Thus, the total time complexity of DRAE is $O(N_{1}N_{2}hdn)$, where $h$ is the number of the hidden layer nodes that can be generally expressed as $d/m$ ($m$ is a constant); then the total time complexity becomes $O(\tilde{C}nd^{2})$ where $\tilde{C}$ is a large constant depending on $N_1$, $N_2$ and $m$.

When the number of points $n$ is large, FastABOD, OCSVM, and DRAE will be very time-consuming. On the contrary, our algorithm takes only linear running time (see Section~\ref{sec-complexity}) and usually runs much faster in practice. For example, our algorithm often takes less than $1/2$ of the time consumed by the other three methods in our experiment.

\section{Extension for Multi-class Inliers}
\label{sec-extension}

All the three compared methods in Section~\ref{sec-exp} can only handle one-class inliers. However, in many real scenarios the data could contain multiple classes of inliers. For example, a given image dataset may contain the images of ``dog" and ``cat", as well as a certain fraction of outliers. So it is necessary to recognize multiple dense regions in the feature space. Fortunately, our proposed algorithm for MEB with outliers can be naturally extended for multi-class inliers. Instead of building one ball, we can perform the following {\em greedy peeling strategy} to extract multiple balls: first we can take a small random sample from the input to roughly estimate the fractions for the classes; then we iteratively run the algorithm for MEB with outliers and remove the covered points each time, until the desired number of balls are obtained. Roughly speaking, we reduce the problem of multi-class inliers to a series of the problems of one-class inliers. The extended algorithm for multi-class inliers is evaluated on two datasets, a random dataset in $\mathbb{R}^{100}$ and Caltech-$256$.

\textbf{Random dataset.} We generate three classes of inliers following different normal distributions and the outliers following uniform distribution in $\mathbb{R}^{100}$. For each outlier ratio $\gamma$, we report the three $F1$ scores (with respect to the three classes of inliers) and their average in Table~\ref{tab7} (a).

\textbf{Caltech-256.} We randomly select three image sets from Caltech-$256$ as the three classes of inliers, and an extra set of mixed images from the remaining image sets as the outliers.
Moreover, we point out that recognizing multi-class inliers from real image sets is much more challenging than single class; we believe that it is due to the following two reasons: (1) the multiple classes of inliers could mutually overlap in the feature space and (2) the outlier ratio with respect to each class usually is large (for example, the outlier ratio for class 1 should also take into account of the fractions of the remaining class 2 and 3, if there are 3 classes in total). We use PCA-VGG-$0.5$ feature in our experiment and the performance is very robust (see Table~\ref{tab7} (b)).

\vspace{-0.3cm}
\newcommand{\tabincell}[2]{\begin{tabular}{@{}#1@{}}#2\end{tabular}}
\begin{table}[htbp]
\centering
\caption{The $F1$ scores of our extended algorithm for multi-class inliers.}
\vspace{2ex}
\begin{tabular}{|c|c|c|c|c|}
\hline
$\gamma$ &$0.1$ &$0.2$ &$0.3$ &$0.4$\\
\hline
$F1$ &\tabincell{c}{$0.970$\\ $0.995$\\$0.994$} &\tabincell{c}{$0.985$\\ $0.995$\\$0.994$} &\tabincell{c}{$0.947$\\ $0.995$\\$0.943$} &\tabincell{c}{$0.995$\\ $0.995$\\$0.963$}\\
\hline
AVG &$0.986$ &$0.991$ &$0.962$ &$0.984$\\
\hline
\end{tabular}\\
\vspace{1ex}
\centering{(a) Random dataset}\\
\vspace{1ex}
\begin{tabular}{|c|c|c|c|c|}
\hline
$\gamma$ &$0.1$ &$0.2$ &$0.3$ &$0.4$\\
\hline
$F1$ &\tabincell{c}{$0.995$\\ $0.993$\\$0.960$} &\tabincell{c}{$0.995$\\ $0.953$\\$0.951$} &\tabincell{c}{$0.995$\\ $0.913$\\$0.968$} &\tabincell{c}{$0.994$\\ $0.928$\\$0.870$}\\
\hline
AVG &$0.983$ &$0.966$ &$0.959$ &$0.931$\\
\hline
\end{tabular}\\
\vspace{1ex}
\centering{(b) Caltech-$256$}
\label{tab7}
\end{table}
\vspace{-0.3cm}
\section{Conclusion}
In this paper, we present a new approach for outlier recognition in high dimension. Most existing methods have high time and space complexities or cannot achieve a quality guaranteed solution. On the contrary, we show that our algorithm yields a nearly optimal solution with the time and space complexities linear on the input size and dimensionality. More importantly, our algorithm can be extended to efficiently solve the instances with multi-class inliers. Furthermore, our experimental results suggest that our approach outperforms several popular existing methods in terms of accuracy.
\bibliographystyle{abbrv}

\bibliography{example_paper}

\begin{thebibliography}{10}

\bibitem{agarwal2005geometric}
P.~K. Agarwal, S.~Har-Peled, and K.~R. Varadarajan.
\newblock Geometric approximation via coresets.
\newblock {\em Combinatorial and Computational Geometry}, 52:1--30, 2005.

\bibitem{agarwal2008robust}
P.~K. Agarwal, S.~Har-Peled, and H.~Yu.
\newblock Robust shape fitting via peeling and grating coresets.
\newblock {\em Discrete \& Computational Geometry}, 39(1-3):38--58, 2008.

\bibitem{aggarwal2001outlier}
C.~C. Aggarwal and P.~S. Yu.
\newblock Outlier detection for high dimensional data.
\newblock {\em ACM Sigmod Record}, 30(2):37--46, 2001.

\bibitem{badoiu2003smaller}
M.~Badoiu and K.~L. Clarkson.
\newblock Smaller core-sets for balls.
\newblock In {\em Proceedings of the ACM-SIAM Symposium on Discrete Algorithms
  (SODA)}, pages 801--802, 2003.

\bibitem{BHI}
M.~Badoiu, S.~Har-Peled, and P.~Indyk.
\newblock Approximate clustering via core-sets.
\newblock In {\em Proceedings of the ACM Symposium on Theory of Computing
  (STOC)}, pages 250--257, 2002.

\bibitem{blum1973time}
M.~Blum, R.~W. Floyd, V.~Pratt, R.~L. Rivest, and R.~E. Tarjan.
\newblock Time bounds for selection.
\newblock {\em Journal of Computer and System Sciences}, 7(4):448--461, 1973.

\bibitem{breunig2000lof}
M.~M. Breunig, H.-P. Kriegel, R.~T. Ng, and J.~Sander.
\newblock Lof: identifying density-based local outliers.
\newblock {\em ACM Sigmod Record}, 29(2):93--104, 2000.

\bibitem{chandola2009anomaly}
V.~Chandola, A.~Banerjee, and V.~Kumar.
\newblock Anomaly detection: A survey.
\newblock {\em ACM Computing Surveys (CSUR)}, 41(3):15, 2009.

\bibitem{C10}
K.~L. Clarkson.
\newblock Coresets, sparse greedy approximation, and the frank-wolfe algorithm.
\newblock {\em ACM Transactions on Algorithms}, 6(4):63, 2010.

\bibitem{Din11}
H.~Ding and J.~Xu.
\newblock Solving the chromatic cone clustering problem via minimum spanning
  sphere.
\newblock In {\em Proceedings of the International Colloquium on Automata,
  Languages, and Programming (ICALP)}, pages 773--784, 2011.

\bibitem{DX14}
H.~Ding and J.~Xu.
\newblock Sub-linear time hybrid approximations for least trimmed squares
  estimator and related problems.
\newblock In {\em Proceedings of the International Symposium on Computational
  geometry (SoCG)}, page 110, 2014.

\bibitem{ding2015random}
H.~Ding and J.~Xu.
\newblock Random gradient descent tree: A combinatorial approach for svm with
  outliers.
\newblock In {\em Proceedings of the AAAI Conference on Artificial Intelligence
  (AAAI)}, pages 2561--2567, 2015.

\bibitem{ester1996density}
M.~Ester, H.-P. Kriegel, J.~Sander, and X.~Xu.
\newblock A density-based algorithm for discovering clusters in large spatial
  databases with noise.
\newblock In {\em Proceedings of the ACM SIGKDD International Conference on
  Knowledge Discovery and Data Mining (KDD)}, pages 226--231, 1996.

\bibitem{fei2007learning}
L.~Fei-Fei, R.~Fergus, and P.~Perona.
\newblock Learning generative visual models from few training examples: An
  incremental bayesian approach tested on 101 object categories.
\newblock {\em Computer vision and Image understanding}, 106(1):59--70, 2007.

\bibitem{frank1956algorithm}
M.~Frank and P.~Wolfe.
\newblock An algorithm for quadratic programming.
\newblock {\em Naval Research Logistics Quarterly}, 3(1-2):95--110, 1956.

\bibitem{GJ09}
B.~G{\"a}rtner and M.~Jaggi.
\newblock Coresets for polytope distance.
\newblock In {\em Proceedings of the International Symposium on Computational
  geometry (SoCG)}, pages 33--42, 2009.

\bibitem{gilbert1966iterative}
E.~G. Gilbert.
\newblock An iterative procedure for computing the minimum of a quadratic form
  on a convex set.
\newblock {\em SIAM Journal on Control}, 4(1):61--80, 1966.

\bibitem{gupta2014outlier}
M.~Gupta, J.~Gao, C.~Aggarwal, and J.~Han.
\newblock Outlier detection for temporal data.
\newblock {\em Synthesis Lectures on Data Mining and Knowledge Discovery},
  5(1):1--129, 2014.

\bibitem{har2004shape}
S.~Har-Peled and Y.~Wang.
\newblock Shape fitting with outliers.
\newblock {\em SIAM Journal on Computing}, 33(2):269--285, 2004.

\bibitem{hinton2012improving}
G.~E. Hinton, N.~Srivastava, A.~Krizhevsky, I.~Sutskever, and R.~R.
  Salakhutdinov.
\newblock Improving neural networks by preventing co-adaptation of feature
  detectors.
\newblock {\em arXiv preprint arXiv:1207.0580}, 2012.

\bibitem{kriegel2009outlier}
H.-P. Kriegel, P.~Kr{\"o}ger, E.~Schubert, and A.~Zimek.
\newblock Outlier detection in axis-parallel subspaces of high dimensional
  data.
\newblock In {\em Proceedings of the Pacific-Asia Conference on Knowledge
  Discovery and Data Mining (PAKDD)}, pages 831--838, 2009.

\bibitem{kriegeloutlier}
H.-P. Kriegel, P.~Kr{\"o}ger, and A.~Zimek.
\newblock Outlier detection techniques.
\newblock {\em Tutorial at PAKDD}, 2009.

\bibitem{kriegel2008angle}
H.-p. Kriegel, M.~Schubert, and A.~Zimek.
\newblock Angle-based outlier detection in high-dimensional data.
\newblock In {\em Proceedings of the ACM SIGKDD International Conference on
  Knowledge Discovery and Data Mining (KDD)}, pages 444--452, 2008.

\bibitem{DBLP:journals/jea/KumarMY03}
P.~Kumar, J.~S.~B. Mitchell, and E.~A. Yildirim.
\newblock Approximate minimum enclosing balls in high dimensions using
  core-sets.
\newblock {\em ACM Journal of Experimental Algorithmics}, 8, 2003.

\bibitem{lecun1998gradient}
Y.~LeCun, L.~Bottou, Y.~Bengio, and P.~Haffner.
\newblock Gradient-based learning applied to document recognition.
\newblock {\em Proceedings of the IEEE}, 86(11):2278--2324, 1998.

\bibitem{liu2014unsupervised}
W.~Liu, G.~Hua, and J.~R. Smith.
\newblock Unsupervised one-class learning for automatic outlier removal.
\newblock In {\em Proceedings of the IEEE Conference on Computer Vision and
  Pattern Recognition (CVPR)}, pages 3826--3833, 2014.

\bibitem{lu2013abnormal}
C.~Lu, J.~Shi, and J.~Jia.
\newblock Abnormal event detection at 150 fps in matlab.
\newblock In {\em Proceedings of the IEEE International Conference on Computer
  Vision (ICCV)}, pages 2720--2727, 2013.

\bibitem{DBLP:journals/corr/Phillips16}
J.~M. Phillips.
\newblock Coresets and sketches.
\newblock {\em Computing Research Repository}, 2016.

\bibitem{rumelhart1988learning}
D.~E. Rumelhart, G.~E. Hinton, and R.~J. Williams.
\newblock Learning representations by back-propagating errors.
\newblock {\em Cognitive Modeling}, 5(3):1, 1988.

\bibitem{sakurada2014anomaly}
M.~Sakurada and T.~Yairi.
\newblock Anomaly detection using autoencoders with nonlinear dimensionality
  reduction.
\newblock In {\em Proceedings of the Workshop on Machine Learning for Sensory
  Data Analysis (MLSDA)}, page~4, 2014.

\bibitem{scholkopf1999support}
B.~Sch{\"o}lkopf, R.~Williamson, A.~Smola, J.~Shawe-Taylor, and J.~Platt.
\newblock Support vector method for novelty detection.
\newblock In {\em Proceedings of the Advances in Neural Information Processing
  Systems (NIPS)}, pages 582--588, 1999.

\bibitem{simonyan2014very}
K.~Simonyan and A.~Zisserman.
\newblock Very deep convolutional networks for large-scale image recognition.
\newblock {\em arXiv preprint arXiv:1409.1556}, 2014.

\bibitem{tan2006introduction}
P.-N. Tan, M.~Steinbach, and V.~Kumar.
\newblock {\em Introduction to Data Mining}.
\newblock 2006.

\bibitem{xia2015learning}
Y.~Xia, X.~Cao, F.~Wen, G.~Hua, and J.~Sun.
\newblock Learning discriminative reconstructions for unsupervised outlier
  removal.
\newblock In {\em Proceedings of the IEEE International Conference on Computer
  Vision (ICCV)}, pages 1511--1519, 2015.

\bibitem{zarrabistreaming}
H.~Zarrabi-Zadeh and A.~Mukhopadhyay.
\newblock Streaming 1-center with outliers in high dimensions.
\newblock In {\em Proceedings of the Canadian Conference on Computational
  Geometry (CCCG)}, pages 83--86, 2009.

\end{thebibliography}


\end{document}